\documentclass{amsart}
\usepackage[utf8]{inputenc}
\usepackage{graphicx}
\usepackage{multicol}
\usepackage{subfigure}
\usepackage{amssymb}
\usepackage{amsmath}
\usepackage{amsthm}

\usepackage{epstopdf}
\usepackage{comment}
\usepackage{hyperref}
\usepackage{gensymb}
\usepackage{enumitem}
\usepackage{tikz}
\usepackage{textgreek}

\usepackage[text={5.5in,8in},centering]{geometry}

\newcommand{\X}{\mathcal{H}_q(\mathbb{F}_{q^2})}

\newtheorem{theorem}{Theorem}
\newtheorem*{theorem*}{Theorem}

\newtheorem{lemma}{Lemma}
\newtheorem{definition}{Definition}
\newtheorem{example}{Example}

\newtheorem{fact}{Fact}

\title{Lower Rate Bounds for Hermitian-Lifted Codes for Odd Prime Characteristic}

\author{Beth Malmskog}
\thanks{This work was supported by the National Science Foundation under Grant No. 2137661}

\author{Na'ama Nevo}
\thanks{}

\begin{document}

\maketitle

\begin{abstract}
Locally recoverable codes are error correcting codes with the additional property that every symbol of any codeword can be recovered from a small set of other symbols.  This property is particularly desirable in cloud storage applications.  A locally recoverable code is said to have availability $t$ if each position has $t$ disjoint recovery sets.  Hermitian-lifted codes are locally recoverable codes with high availability first described by Lopez, Malmskog, Matthews, Pi\~nero-Gonzales, and Wootters. The codes are based on the well-known Hermitian curve and incorporate the novel technique of lifting to increase the rate of the code. Lopez et al. lower bounded the rate of the codes defined over fields with characteristic 2. This paper generalizes their work to show that the rate of Hermitian-lifted codes is bounded below by a positive constant depending on $p$ when $q=p^l$ for any odd prime $p$.
\end{abstract}

\section{Introduction}

Error-correcting codes are methods to cleverly encode redundancy into information so that errors or erasures occurring during transmission or storage can be repaired. These algorithms are ubiquitous, enabling processes as varied and important as telescopes transmitting images of the universe back to earth, safe storage of photos on personal laptops, and reliable cloud storage and computing.  

This paper focuses particularly on locally recoverable codes, a class of codes motivated by cloud storage applications.   Intuitively, locally recoverable codes have the additional property that any erasure can be recovered with access to only a small number of other symbols from the encoded information, called the recovery set for the erased symbol.  In cloud storage applications, individual servers will not infrequently fail or become unavailable.  To make sure that storage is reliable, one piece of information (or codeword) may be spread across many servers so that if a single server is not available, the information of this server can be retrieved from other servers in the cluster.  This eliminates the need for simple backups of each server, which would be necessary if the entire codeword were stored on a single server.  In a situation of high demand or high number of failures, it would be nice to have multiple ways to recover the information from a single server.  A locally recoverable code is said to have availability $t$ if for each server there are $t$ mutually disjoint sets of servers that can each be used to recover the information of the original server. 

This work builds on that of Lopez, Malmskog, Matthews, Pi\~nero-Gonzales, and Wootters\cite{malmskog:article} constructing locally recoverable codes with high availability called Hermitian-lifted codes. These are evaluation codes on the affine points of the Hermitian curve $y^{q+1}=x^q+x$ defined over the finite field $\mathbb{F}_{q^2}$, where $q=p^l$ for any prime $p$ and any positive integer $l$. The rate of a code is a measure if its efficiency, the ratio of the number of symbols in the raw data to the number of symbols in the encoded information, so codes with larger rate are more efficient.  Their paper proves that when $q=2^l$, the rate of the Hermitian-lifted code is bounded below by a positive constant independent of $l$. This work uses the same methods to show that when $q=p^l$ for any prime $p$, the rate of the Hermitian-lifted code is still bounded below by a positive constant depending on $p$ but independent of $l$. Although we find a lower bound that is fairly small, this result is important because it shows that the rate of this code is significantly better than the rate of the corresponding Hermitian one-point code $\mathcal{C}$, which tends to $0$ for $q=p^l$ as $l$ increases. This is further evidence that applying the lifting techniques pioneered by Guo et al. \cite{guo2013new} to evaluation codes on curves can be effective for odd primes as well.  These ideas have promise for other curves, as illustrated by the recent work of Matthews et al. applying similar lifting techniques to Norm-Trace curves in \cite{matthews_murphy} and \cite{matthews2023curvelifted}. We note that the lower bound found here is not a tight bound. The bound calculated in the theorem is based on finding enough ``good" functions to yield a positive rate bound, but does not aim to find all the ``good" functions or the actual dimension of the code.  In fact, recent work of Allen et al. \cite{undergrad:future} (completed after this work) improves upon our bound using different techniques.  However, an important aspect of our work is that it illustrates that the techniques of \cite{malmskog:article}, which could be potentially useful for other curves, can be extended to general primes, and proves some elegant but technical number theoretic identities for all primes $p$.

Section \ref{sec:background} will cover the necessary background on codes and algebraic geometry codes, and then introduce the construction for Hermitian-lifted Codes described in \cite{malmskog:article}. In Section \ref{sec:proof}, we will prove the generalized version of the theorem proved in \cite{malmskog:article} by following a very similar proof structure to the one used to prove the $q=2^l$ case. Finally, we will show an example of the proof in the $p=3$ case and an example of a good monomial.

\section{Important Background and Notation}\label{sec:background}

In the following, we include standard definitions and facts from coding theory and algebraic geometry codes.  For more information and a full development of this background material, a good reference is Judy Walker's \textit{Codes and Curves} \cite{walker:textbook}.

\subsection{Error Correcting and Detecting Codes}

Let $p$ be an integer prime, $l$ a natural number, and $q=p^l$.  Let $\mathbb{F}_q$ denote the field with $q$ elements. For two natural numbers $n,k$ with $n\geq k$, a \textit{linear code} $C$ of \textit{length} $n$ and \textit{dimension} $k$ is a $k$-dimensional linear subspace of $(\mathbb{F}_q)^n$.  Information is encoded as vectors of length $n$ with coordinates in $\mathbb{F}_q$.  Elements of $C$ are called \textit{codewords}. The code $C$ has $q^k$ codewords, i.e. there are $q^k$ different messages or words that can be transmitted or stored using $C$.

For any two vectors in $(\mathbb{F}_q)^n$,  $\vec{x}=(x_1,x_2,\ldots, x_n)$ and $\vec{y}=(y_1,y_2,\ldots, y_n)$, the \textit{Hamming Distance} between any $\vec{x}$ and $\vec{y}$ is defined as $d(\vec{x},\vec{y})=\#\{i|x_i\neq y_i\}$ , which is the number of positions in which $\vec{x}$ and $\vec{y}$ have differing symbols. Every code has a \textit{minimum distance} usually denoted $d$, which is the smallest Hamming Distance between any two distinct codewords in the code. The \textit{minimum weight} of a code is the smallest distance between a nonzero codeword and the zero codeword.  It is not hard to see that for a linear code $C$, the minimum weight is equal to the minimum distance. The minimum distance of a code determines how many errors in a message the code can correct. Say that codeword $\vec{c}$ is transmitted and stored, but some number $e$ of positions are corrupted.  If $e<d$, the resulting message is not a codeword, because it would take at least $d$ changes to get from $\vec{c}$ to the nearest other codeword.  Thus up to $d-1$ errors can be \textit{detected}.  If $e\leq\frac{d-1}{2}$, the message will be closer to $\vec{c}$ than to any other codeword, so up to $\frac{d-1}{2}$ errors can be \textit{corrected} by replacing the corrupted message with the closest codeword.

The \textit{information rate} of a code is given by $R=k/n$. Intuitively, out of the $n$ symbols in a codeword, only $k$ symbols provide are necessary to convey the message, while the other $n-k$ symbols are transmitted to assist with error detection. Therefore, maximizing the rate of a code minimizes the number of extraneous symbols that need to be transmitted, which increases the efficiency of transmission. A rate as close as possible to $1$ is ideal.

Although maximizing both the minimum distance and the dimension of a code is ideal, when the two values are maximized there must be trade-offs between them.  This is captured in the \textit{Singleton bound}.

\begin{theorem}
    For a linear code with minimum distance $d$, dimension $k$, and length $n$, we have $d\leq n-k+1$.
\end{theorem}

A linear code with length $n$ and dimension $k$ can be represented by a $k\times n$ generator matrix, where each row of the matrix is one of the $k$ basis elements of the code's vector space.

\begin{example}
\normalfont
Consider the linear code $\mathcal{C}$ over the alphabet $A=\mathbb{F}_2=\{0,1\}$ with the following generator matrix.
$$\begin{bmatrix}
1 & 0 & 0 & 1 & 1 & 1\\
0 & 1 & 0 & 0 & 1 & 1\\
1 & 0 & 1 & 0 & 0 & 1\\
\end{bmatrix}$$

It is easy to see that the code $\mathcal{C}$ has length $n=6$ and dimension $k=3$. Note that the minimum distance of the set of basis vectors is not equal to the minimum distance of the entire code, which often makes it very difficult to calculate the minimum distance of large codes. Since this example is small and only has $8$ codewords, we can use brute force to find the minimum distance. Writing out and comparing all the elements of the code shows that the minimum distance of $\mathcal{C}$ is $3$.

Now consider that the message $(1,1,0,1,0,1)$ is received. Since this vector is not a linear combination of any of the rows of the generator matrix, then the message cannot be a codeword in $\mathcal{C}$. Therefore, an error must have occurred in the transmission.

A minimum distance of $3$ means this code can correct at most $\frac{3-1}{2}=1$ error in a codeword. Thus, in order to correct the error, we must find a codeword in $\mathcal{C}$ that has a Hamming Distance of $1$ from the received message. The codeword $(1,1,0,1,0,0)$ is generated by adding the first two rows of the matrix, and has a Hamming Distance of $1$ from the received message. Then this codeword is the closest to the message, so the code algorithm would assume that the intended message was $(1,1,0,1,0,0)$.
\end{example}
   We say that $C$  is a \textit{locally recoverable code (LRC) with locality $r$} if for each $i\in \{1,\dots,n\}$ there exists a set of $r$ indices $A_i\subseteq\{1,\dots,n\}\setminus\{i\}$ with $\#A_i = r$ so that, for any codeword $\vec{c}=(c_1,\dots , c_n) \in C$ the value $c_i$ can be recovered using the values in the positions of $A_i$.
    The set $A_i$ is called the \textit{recovery set} for the $i$-th position.  One of the first recovery schemes of this nature appeared in \cite{huang2013pyramid}.  
    A locally recoverable code $C$ has \textit{availability} $t$ with locality $(r_1, \dots, r_t)$ if for each $i\in \{1,\dots,n\}$ there exist $t$ disjoint recovery sets for position $i$.  When $r_i=r_j=r$ for all $i,j\in\{1,2,\dots, t\}$, we say that the code has \textit{uniform locality} $r$, or simply locality $r$.

\subsection{Evaluation Codes}
Let $S=\{P_1,P_2,\dots, P_n\}$ be a finite set and $L$ an $\mathbb{F}_q$-linear space of functions defined on $S$ with values in $\mathbb{F}_q$.  An \textit{evaluation code} $C(S,L)$ is defined by evaluating the functions in $L$ on the points of $S$:
\[C(S,L)=\{(f(P_1), f(P_2), \dots, f(P_n)): f\in L\}.\]  If the evaluation map of $L$ on $S$ is injective, then $C(S,L)$ is isomorphic to $L$ as a vector space, so has the same dimension as $L$. If $S$ is a subset of the points on an algebraic variety $\mathcal{X}$, the geometric and algebraic structure of $\mathcal{X}$ can often give valuable information about the parameters of $C(S,L)$. This underlying structure can also naturally give rise to locality and availability.  Important examples of evaluation codes are Reed-Solomon codes, Reed-Muller codes, and one-point codes on curves. 

\begin{example}
    Let $S=\{P_1, P_2,\dots, P_q\}$ be the set of all values in $\mathbb{F}_q$, and let $L=L_{k-1}$ for some $k\leq q$ be the set of all polynomials in $\mathbb{F}_q[x]$ of degree at most $k-1$.  Then $C(S,L)$ is called a \textit{Reed-Solomon code} of dimension $k$ over $\mathbb{F}_q$, denoted $RS(q,k)$.  It is not hard to check that this code has length $q$, dimension $k$, and minimum distance $d=q-k+1$.  Thus Reed-Solomon codes meet the Singleton bound, and are in that sense as good as possible!
\end{example}

Reed-Solomon codes are excellent codes that have been used in many settings, including encoding music on compact discs and encoding information for transmission to Earth from the Voyager space probe. We can think of these codes geometrically by seeing the evaluation set as the points on a ``number line" over $\mathbb{F}_q$. In one very influential construction, Tamo and Barg \cite{TamoBargOriginal} devised local recovery methods for certain subcodes of Reed-Solomon codes.

However, the dimension and length of Reed-Solomon codes are by the size of the field $\mathbb{F}_q$.  Because arithmetic becomes more time consuming as the field size grows, we may wish for longer codes over a fixed size field, while keeping the benefits of geometric constructions.  Reed-Muller codes generalize Reed-Solomon codes by growing the evaluation set from a line to an $r$-dimensional space.

\begin{example}
    Let $m$ be a natural number.  Let $S=\{P_1, P_2,\dots, P_{q^m}\}$ be the set of points in the $m$-dimensional space $(\mathbb{F}_q)^m$.  For $r$ a non-negative integer, let $L=L_{r,m}$ be the set of $m$-variate polynomials in $\mathbb{F}_q[x_1,x_2,\dots, x_m]$ of total degree at most $r$, where the degree in $x_i$ is at most $q-1$.  Then $C(S,L)$ is called a \textit{($q$-ary) Reed-Muller code} of degree $r$ in $m$ variables and denoted $RM_q(r,m)$.  The code $RM_q(r,m)$ has length $n=q^m$, dimension $k=\binom{r+m}{m}$, and minimum distance $d=(q-r)q^{r-1}$ when $r\leq q-1$.
\end{example}

Reed-Muller codes with $r\leq q-2$ are locally recoverable with locality $r+1$ and availability $\frac{q^m-1}{q-1}$.  The idea behind local recovery is the following.  If $f\in L$ has total degree $\leq q-2$, then the function $f$ restricted to any line in $(\mathbb{F}_q)^m$ will be a univariate polynomial $\tilde{f}$ of degree $\leq q-2$.  If position $i$ in a codeword is erased, we need to recover the value of $f(P_i)$ where $P_i\in S$.  By choosing any line through $P_i$, we can say that $f(P_i)=\tilde{f}(t_i)$ for some $t_i\in\mathbb{F}_q$, where the values of $\tilde{f}(t)$ are known for the other $q-1$ points on the line, corresponding to $q-1$ distinct values of $t$.  Using Lagrange interpolation, we can determine the coefficients of a degree $q-2$ polynomial with these $q-1$ data points.  Thus every line passing through the point $P_i$ gives a recovery set for position $i$.  Reed-Muller codes are also widely studied, and have applications beyond cloud storage, including private information retrieval.  

Another way of generalizing the Reed-Solomon code to increase the length is to use the points on a curve in the plane as the evaluation set, instead of line.  A \textit{plane curve} can be a defined as the set of solutions to a non-trivial polynomial equation in two variables.  The \textit{Hermitian curve} $\mathcal{H}_q$ can be defined by the equation $x^q+x+y^{q+1}=0$, and we often search for points over the field $\mathbb{F}_{q^2}$. Note that the Hermitian curve is defined given by an equation of the form $x^q+x-y^{q+1}=0$, but these are equivalent under the change of variables $(x,y)\mapsto (-x,y)$, and the given form simplifies the algebra in our proof. The Hermitian curve has many beautiful properties, including an exceptional number of symmetries and as many points as possible for a curve of its complexity over a $\mathbb{F}_{q^2}$.  Goppa introduced algebraic geometry (AG) codes using Riemann-Roch spaces in \cite{goppa1982algebraico}.  We do not use the full generality here, but remark that these AG codes have provided remarkable examples of long codes with desirable parameters. See \cite{walker:textbook} for a discussion.  We also note that Tamo, Barg, and Vladut also defined locally recoverable codes and codes with availability on the curves more generally using ideas from covering maps \cite{BTV}, but the construction here is instead based on intersections. We will describe the most relevant example for our work, a Hermitian one-point code.  The term ``one-point" refers to the fact that this is an evaluation code where the evaluated functions are only allowed to have poles at the single ``point at infinity" on the curve. Hermitian-lifted codes are closely related to Hermitian one-point codes.
\begin{example}
 Let $S=\{P_1,P_2,\dots, P_{q^3}\}$ be the set of pairs $(x,y)\in(\mathbb{F}_{q^2})^2$ satisfying $x^q+x+y^{q+1}=0$.  For future reference, we will use $\mathcal{H}_q(\mathbb{F}_{q^2})$ to denote this set.  Using the properties of the field norm and trace maps $\mathbb{F}_{q^2}\rightarrow \mathbb{F}_q$, it is not hard to see that there are $q^3$ such pairs.  For a non-negative integer $m$, let $L=L_m$ be the $\mathbb{F}_{q^2}$-linear space of functions generated by the set \[\{x^iy^j:0\leq j\leq q-1,iq+j(q+1)\leq r\}.\]  We then define the Hermitian one-point code to be $C(S,L)$, and denote this code by $C_{q,r}$.  The Hermitian one-point code has length $n=q^3$.  If $q^2-q-2< r < q^3$, the code has dimension $k=r+1-\frac{q(q-1)}{2}$.  The minimum distance satisfies $d\geq n-r$ when $r<q^3$ \cite[Proposition 8.3.3]{stichtenoth2009algebraic}.
\end{example} 

Using the same idea as in Reed-Muller codes, Hermitian one-point codes are also locally recoverable with high availability.  To see why, we note that an exceptional property of the Hermitian curve is that if you work in projective space and include a point at infinity on the curve, then every non-tangent line to the curve in projective space intersects in exactly $q+1$ points, where as tangent lines intersect the curve in exactly 1 point.  Since we do not wish to include a full discussion of projective space here (again, see \cite{walker:textbook} for more discussion), we will make use of this custom version.

\begin{fact}\cite{malmskog:article} Every non-horizontal line in the plane intersects $\mathcal{H}_q$ in either 1 or $q+1$ points.
\end{fact}

 In the Hermitian one-point code, similar to the Reed-Muller case, the recovery sets arise from the intersections of lines with the curve, where the value of a function at any point on a line can be recovered by the values of the function at the remaining points on the line.

\begin{theorem}\cite{malmskog:article} Hermitian one-point code $C_{q,q^2-1}$ has locality $q$ and availability $q^2-1$.
\end{theorem}
\begin{proof}
Let each index $i$ correspond to a point $P_i$ in $S=H_q(\mathbb{F}_{q^2})$. For any $\alpha,\beta\in\mathbb{F}_{q^2}$, let $L_{\alpha,\beta}:\mathbb{F}_{q^2}\rightarrow (\mathbb{F}_{q^2})^2$ so that $L_{\alpha,\beta}(t)=(\alpha t+\beta,t)$. Let $\mathcal{L}_{\alpha,\beta}$ be the image of this map.  For any line passing through $P_i$ that is not tangent, then let $R_{i,\alpha}=H_q(\mathbb{F}_{q^2}\cap \mathcal{L}_{\alpha,\beta}\setminus \{P_i\})$. Note that $|R_{i,\alpha}|=q$ and there are $q^2-1$ disjoint sets for each $i$.

Any codeword $\vec{c}_f$ in $C_{q,q^2-1}$ arises from the evaluation of a polynomial $f$ that is a linear combination of $x^ay^b$ satisfying $b\leq q-1$ and $aq+b(q+1)\leq q^2-1$. Using these inequalities and the fact that $a$ and $b$ are integers, this implies that
$$a+b\leq q-1.$$

Since the total degree of $f$ is at most $q-1$, we have that $f$ restricted to $\mathcal{L}_{\alpha,\beta}$ is a univariate polynomial $\tilde{f}$ of degree at most $q-1$. Given $q$ values of $\tilde{f}$ from the other $q$ points on $\mathcal{L}_{\alpha,\beta}$, we can interpolate the value of $\tilde{f}$ at $P_i$. 
\end{proof}

\section{The Hermitian-Lifted Code}

The high availability and long length relative to field size make the Hermitian one-point code appealing to study. However, the rate of the Hermitian one-point code $C_{q,q^2-1}$ is 
\[\frac{(q^2-1)+1-\frac{q(q-1)}{2}}{q^3}=\frac{2q-1}{2q^2},\] which approaches 0 as $q$ grows.  The Hermitian-lifted code shares many properties with the Hermitian one-point code, but we will prove that it has the advantage of a higher rate. The Hermitian-lifted code has the same evaluation points the Hermitian one-point code, but it extends the set of functions that can be used to form codewords. The idea comes from lifted Reed-Solomon codes, first introduced by Guo, Kopparty, and Sudan \cite{guo2013new}. Their lifted Reed-Solomon codes can be thought of as Reed-Muller codes with additional functions.  Instead of only using functions with low total degree, lifted Reed-Solomon codes include higher degree polynomials with the special property that, when restricted to each line in space, the restricted function has low enough degree that missing values can be interpolated from remaining values corresponding to other points on the line.  The surprising insight of \cite{guo2013new} is that there are many such polynomials, enough to greatly increase the rate of the lifted Reed-Solomon code over the standard Reed-Muller code. 

The definitions in this section come directly or are modified from \cite{malmskog:article}, where the Hermitian-lifted code was first described.

\begin{definition}
For $L_{\alpha,\beta}=(\alpha t+\beta,t)$ $\alpha,\beta\in \mathbb{F}_{q^2}$, $f\in\mathbb{F}_{q^2}[x,y]$ and $g\in\mathbb{F}_{q^2}[t]$, we say that $f$ agrees with $g$ on the intersection of $H_q(\mathbb{F}_{q^2})$ and $L_{\alpha,\beta}$ if $f(L_{\alpha,\beta}(t))=g(t)$ for all $t\in\mathbb{F}_{q^2}$ with $L_{\alpha,\beta}(t)\in \mathcal{H}_q(\mathbb{F}_{q^2})$.
\end{definition}

\begin{definition}
Given a prime power $q$, let $\mathcal{L}=\{L_{\alpha,\beta}:\alpha,\beta\in\mathbb{F}_{q^2}\}$ is the set of all lines of the form $L_{\alpha,\beta}(t)=(\alpha t+\beta, t)$.  Let $\mathcal{F}$ be the set of all $f\in\mathbb{F}_{q^2}[x,y]$ such that for each $L\in\mathcal{L}$ there exists $g\in\mathbb{F}_{q^2}[t]$ so that $\deg(g)\leq q-1$ and $f$ agrees with $g$ on the intersection of $\X$ and $L$.

\end{definition}

In other words, $\mathcal{F}$ is the set of all functions $f$ such that for every line $L$, plugging in $f(L(t))$ yields a univariate polynomial that restricts to degree at most $q-1$ on the points of intersection between the curve and the line.  

Now we define Hermitian-lifted codes.

\begin{definition}
Let $q$ be a prime power and let $\mathcal{F}$ be defined as above. Then the Hermitian-lifted code $\mathcal{C}_q$ is the evaluation code
$C(\mathcal{H}_q(\mathbb{F}_{q^2}),\mathcal{F})$.
\end{definition}

The length of $\mathcal{C}_q$ is $q^3$ because $|\mathcal{H}_q(\mathbb{F}_{q^2})|=q^3$. Each codeword is the vector created by evaluating on all points of the Hermitian curve a function $f$ that reduces to degree at most $q-1$ on every line intersected with the Hermitian curve.

The Hermitian one-point code $C_{q,q^2-1}$ and the Hermitian-lifted code $\mathcal{C}_q$ are closely related, so we restate that the difference between them is simply in the functions $f$ that are used in each. In the one-point code $C_{q,q^2-1}$, the functions $f$ must satisfy $f\in L_{q^2-1}$, which is a subset of the two variable polynomials of total degree at most $q-1$. In the Hermitian-lifted code, the functions must satisfy $f\in \mathcal{F}$, which is all functions of unbounded total degree but having degree at most $q-1$ when restricted to the intersection of the curve and each non-horizontal line. Although the distinction between these two codes is seemingly insignificant, the main theorem proves that the slight difference in the set of functions $\mathcal{F}$ and $L_{q^2-1}$ is enough to bound the rate of the Hermitian-lifted code away from 0 as $q$ goes to infinity.

\section{Main Theorem and Proof}\label{sec:proof}

The authors of \cite{malmskog:article} prove that the rate of the Hermitian-lifted code as $q\rightarrow \infty$ when $q=2^l$ is bounded below by a positive constant. 

\begin{theorem}\cite{malmskog:article}\label{main2}
Suppose $q=2^l$ where $l\geq 2$ and $\mathcal{C}_q$ be the Hermitian-lifted code. Then the rate of $\mathcal{C}_q$ is at least $0.007$.
\end{theorem}

The aim of this paper is to extend this result to show that the rate of the Hermitian-lifted code is also bounded below by a positive constant when $q=p^l$ when $p$ is any odd prime. Despite the fact that $0.007$ is a very small lower bound, it demonstrates a significant difference between the Hermitian-lifted code and the Hermitian one-point code, which has a rate that converges to $0$ as $q$ increases. 

Finding that the rate of the Hermitian-lifted code is bounded below by a number greater than $0$ proves that the set of functions $\mathcal{F}\setminus L_{q^2-1}$ is large. The proof of Theorem ~\ref{main2} in \cite{malmskog:article} follows the appoach of finding a large set of functions included in the lifted code but not in the one-point code. We adopt the same techniques to prove the following main result.

\begin{theorem*}
   Suppose that $q=p^l$ where $p$ is an odd prime and $l\geq 2$. Then the rate of $\mathcal{C}_q$ is at least $$\frac{.469}{p^4(p-1)(p^3-p^2+1)}.$$
\end{theorem*}

Although the bound decreases as $p$ increases, the bound always remains positive. The remainder of this section will be the proof to the above Theorem, also referred to as Theorem ~\ref{main}. First, we will describe the set of functions that will provide a lower bound on the dimension of the code, a set which we call \textit{good monomials}. Then, we will count the minimum number of functions that must be good monomials, which will lead to a lower bound on a dimension. The final proof of the bound in Theorem ~\ref{main} will conclude this section.

\subsection{Good Monomials}

We begin with a lemma from \cite{malmskog:article} giving a large set of functions on $\mathcal{H}_q$ leading to linearly independent evaluation vectors. Since this lemma is proven for all $p$ in \cite{malmskog:article}, we omit a proof here. 
\begin{lemma}
\cite{malmskog:article} Let $q=p^l$ for $p$ any prime, $l$ a natural number.  Let $M_{a,b}(x,y)=x^ay^b$. Then the set of vectors $\{(M_{a,b}(P_i))_{P_i\in \X}:0\leq a \leq q-1,0\leq b\leq q^2-1\}$ are linearly independent.
\end{lemma}
 
Since these monomials give linearly independent evaluation vectors, any monomial in this set that is also in $\mathcal{F}$ will contribute to the dimension of the code.  Before we can determine which monomials are fit this description, we will introduce another definition. Given a line of the form $L_{\alpha,\beta}(t)=(\alpha t+\beta,t)$, we find a polynomial with zeros at exactly the points of intersection between the Hermitian curve and the line.  If $M_{a,b}(L_{\alpha,\beta}(t))$ has degree at most $q-1$ when reduced modulo this polynomial for all $\alpha, \beta\in\mathbb{F}_{q^2}$, then $M_{a,b}$ is good.  The folowing definition generalizes the corresponding statement in \cite{malmskog:article}.

\begin{definition} For any prime power $q$, let $h(x,y)=y^{q+1}+x^q+x$. For any $\alpha,\beta\in\mathbb{F}_{q^2}$, define 
\[p_{\alpha,\beta}(t)=h(\alpha t+\beta,t)=t^{q+1}+\alpha^qt^q+\alpha t+(\beta+\beta^q)=t^{q+1}+a^qt^q+\alpha t+\gamma.\]
For $g(t)\in\mathbb{F}_{q^2}[t]$, $\bar{g}(t)$ is the remainder when $g(t)$ is divided by $p_{\alpha,\beta}(t)$.
Let $deg_{\alpha,\beta}(g)=deg(\bar{g}_{\alpha,\beta}(t))$. 
\end{definition}
Note that $\deg_{\alpha,\beta}(g)\leq q$ for all $g\in\mathbb{F}_{q^2}[t]$ because $\deg(p_{\alpha,\beta}(t))=q+1$.

For a line $L_{\alpha,\beta}(t)=(\alpha t+\beta,t)$, we have $M_{a,b}(L_{\alpha,\beta}(t))$ agrees with polynomial $g$ of degree strictly less than $q$ on $\X$ if and only if $\text{deg}_{\alpha,\beta}(M_{a,b}(L_{\alpha,\beta}(t))<q$. Write \[M_{a,b}(L_{\alpha,\beta}(t))=h(t)p_{\alpha,\beta}(t)+g(t)\] for $\deg(g)\leq q$. The roots of $p_{\alpha,\beta}(t)$ are exactly the $t$-values in the parameterization of the line $L_{\alpha,\beta}$ which intersect the Hermitian curve, and these roots satisfy $0 = t^{q+1}+(\alpha t+\beta )^q+\alpha t+\beta $ so $p_{\alpha,\beta}(t)=0$. Thus, $M_{a,b}(L_{\alpha,\beta}(t))$ agrees with $g(t)$ on the intersection of $\X$ and the line $L_{\alpha,\beta}$. 
\begin{definition}
    For a monomial $M_{a,b}(x,y)$, let $g_{a,b}(t)=M_{a,b}(L_{\alpha,\beta}(t))=(\alpha t+\beta)^at^b$. We define a \textit{good monomial} to be any $M_{a,b}(x,y)$ such that $0\leq a\leq q-1$, $0\leq b \leq q^2-1$, and $g_{a,b}(t)$ satisfies $\deg_{\alpha,\beta}(g_{a,b})\leq q-1$ for all $\alpha, \beta\in\mathbb{F}_{q^2}$.
\end{definition}

\subsection{Conditions of Good Monomials}

In what follows, we assume that the line $L_{\alpha,\beta}(t)=(\alpha t+\beta,t)$ is a line that goes through $q+1$ points of the Hermitian curve, where $\alpha, \beta\in\mathbb{F}_{q^2}$. We ignore tangent lines because all monomials agree with to constant functions on a single point of intersection.  Also, we let $\gamma=\beta+\beta^q\in\mathbb{F}_q$.  For simplicity, we let $p(t)=p_{\alpha,\beta}(t)$.

Let $\sigma_0,\ldots,\sigma_q$ be the roots of $p(t)$. We know there are $q+1$ roots because there are $q+1$ points in the intersection of the line and the Hermitian curve. Thus,
\[p(t)=t^{q+1}+\alpha^qt^q+\alpha t+\gamma=(t-\sigma_0)\cdots(t-\sigma_q)=c_0t^{q+1}+c_1t^q+\cdots+c_qt+c_{q+1}\]
where $c_k=\sum_{S\subset\{0,\ldots,q\},|S|=k}\prod_{l\in S}\sigma_l$ for $k=0,\ldots,q$. This is just given by the expansion of the product above.

For any $k\geq 0$ we define the element $P_k=\sum_{i=0}^q\sigma_i^k$. These values will be used to find a condition for good monomials. The statements and proofs of Lemmas \ref{pkcondtion} and \ref{pkrule} below are very similar to Proposition $7$ and Lemma $8$ in \cite{malmskog:article}, respectively, with variations resulting from our focus on odd primes $p$ instead of only $p=2$.

\begin{lemma}\label{pkcondtion}
Let $q$ be a power of $p$ an odd prime, and let $\alpha,\beta\in\mathbb{F}_{q^2}$. Then $P_{k+1}=-\alpha^qP_k$ if and only if $deg_{\alpha,\beta}(t^k)<q$.
\end{lemma}

\begin{proof}
Divide $t^k$ by $p(t)$ to yield $t^k=g_k(t)p(t)+\bar{g}_k(t)$ for some polynomial $\bar{g}_k(t)$ of degree at most $q$. We will show that $\deg(\bar{g}_k(t))<q$ exactly when $P_{k+1}=-\alpha^qP_k$.

We may attain $q+1$ values of $\bar{g}_k(t)$ by noting that $\bar{g}_k(\sigma_i)=\sigma_i^k$ for $0\leq i\leq q$. This is true because each $\sigma_i$ is a root of $p(t)$.  Since $\bar{g}_k$ has degree less than $q$, Lagrange interpolation yields
\begin{equation}\label{laginterp}
    \bar{g}_k(t)=\sum_{i=0}^k\sigma_i^k\prod_{j\neq i}\left(\frac{t-\sigma_j}{\sigma_i-\sigma_j}\right)=\left(\sum_{i=0}^q\sigma_i^k\prod_{j\neq i}\frac{1}{\sigma_i-\sigma_j}\right)t^q+r(t)
\end{equation}
where $\deg(r(t))<q$. Since we are only concerned with checking when the degree of the whole polynomial is less than $q$, it is enough to check which conditions ensure that the coefficient of $t^q$ is $0$. We will use an identity arising from the derivative of $p(t)$. Starting with
$$p(t)=t^{q+1}+\alpha^qt^q+\alpha t+\gamma=(t-\sigma_0)\cdots(t-\sigma_q)$$
we can take the derivative of both sides and get
$$p'(t)=t^q+\alpha=\sum_{i=0}^q\prod_{j\neq i}(t-\sigma_j).$$
Replacing $t$ with $\sigma_i$ yields
\begin{equation}\label{plugin}
    p'(\sigma_i)=\sigma_i^q+\alpha=\prod_{j\neq i}(\sigma_i-\sigma_j).
\end{equation}
Because $\sigma_i$ is a root of $p(t)$ then $\sigma_i^{q+1}+\alpha^q\sigma_i^q+\alpha\sigma_i=-\gamma$.  We add both sides to $\alpha^{q+1}$ and factor to get $(\sigma_i^q+\alpha)(\sigma_i+\alpha^q)=\alpha^{q+1}-\gamma$.  Therefore 
\[\prod_{j\neq i}(\sigma_i-\sigma_j)=\frac{\alpha^{q+1}-\gamma}{\alpha^q+\sigma_i}.\]
Using equation \ref{laginterp} we can calculate the coefficient of $t^q$ in $\bar{g}_k(t)$ to be $$\sum_{i=1}^q\frac{\sigma_i^k(\alpha^q+\sigma_i)}{\alpha^{q+1}-\gamma}=\frac{\alpha^qP_k+P_{k+1}}{\alpha^{q+1}-\gamma}$$
which is equal to zero exactly when $P_{k+1}=-\alpha^qP_k$. Thus, this is our condition for when $deg_{\alpha,\beta}(g)<q$.
\end{proof}

Now that we have a sufficient condition that gives powers of $t$ that reduce sufficiently modulo $p(t)$, we need to find a condition on $k$ that will yield $P_{k+1}=-\alpha^qP_k$. We will find a few patterns to assist.

\begin{lemma}\label{pkrule}
Let $q$ be a power of an odd prime $p$. For $0\leq k<q$, $P_k=(-1)^k\alpha^{qk}$ and $P_{kq}=(-1)^k\alpha^k$.
\end{lemma}

\begin{proof}
The fact that we are working in the finite field $\mathbb{F}_{q^2}$ implies $P_0=q+1=1\in\mathbb{F}_{q^2}$. Let $1\leq k<q$. We get $P_k=-\alpha^qP_{k-1}$ from Lemma \ref{pkcondtion}. By induction, we get that $P_k=(-1)^k\alpha^{qk}$. Again, because we are working over $\mathbb{F}_{q^2}$, we can use the rules of finite field arithmetic to find
$$P_{kq}=\sum_{i=0}^q\sigma_i^{qk}=\left(\sum_{i=0}^q\sigma_i^k\right)^q=((-1)^k\alpha^{qk})^q=(-1)^ k\alpha^k.$$
Therefore, $P_k=(-1)^k\alpha^{qk}$ and $P_{kq}=(-1)^k\alpha^k$.
\end{proof}

Using Lemmas ~\ref{pkcondtion} and ~\ref{pkrule} above, we can find a relationship between the elements of the following matrix, which will be useful for the upcoming theorems.
$$\begin{pmatrix}
P_0 & P_q & \cdots & P_{(q-1)q} \\
P_1 & P_{q+1} & \cdots & P_{(q-1)q+1} \\
\vdots & \vdots & & \vdots \\
P_{q-1} & P_{2q-1} & \cdots & P_{q^2-1}
\end{pmatrix}.$$
For a root $\sigma$ of $p(t)=t^{q+1}+\alpha^qt^q+\alpha t+\gamma$, we have $-\sigma^{q+1}=\alpha^q\sigma^q+\alpha\sigma+\gamma$. By multiplying both sides of the equation by $\sigma^{k-q-1}$ we get $-\sigma^k=\alpha^q\sigma^{k-1}+\alpha\sigma^{k-q}+\gamma\sigma^{k-q-1}$. Thus, summing over $0\leq i\leq q$, we obtain that the values of $P_k$ satisfy the recurrence relation
\begin{equation}\label{relation}
    P_k=-\alpha^qP_{k-1}-\alpha P_{k-q}-\gamma P_{k-q-1}.
\end{equation}
Based on this formula, the $(i,j)$ entry of the matrix is determined by the $(i-1,j)$, $(i,j-1)$, and $(i-1,j-1)$ entries of the matrix. As a result, the entire matrix is determined by the first row and first column of the matrix. Therefore, every $2\times 2$ submatrix $M$ must satisfy the recurrence relation
\begin{equation}\label{blockrelation}
    M_{22}=-\alpha^qM_{12}-\alpha M_{21}-\gamma M_{11}.
\end{equation}
The next step of the proof will be to show that the above matrix can be written as a product of a specific product of matrices, using the fact that it is sufficient to show that the first row and column are equal and every submatrix satisfies the recurrence relation. The following definitions will provide the tools for defining the product.

\begin{definition}
Let $A=[a_{ij}]$ be an $r\times s$ matrix and $B=[b_{ij}]$ an $m_1\times m_2$ matrix. The Kronecker product of $A$ and $B$ is the $rm_1\times sm_2$ matrix that can be expressed in block form as $$A\otimes B=\begin{pmatrix}
a_{11}B & a_{12}B & \cdots & a_{1s}B \\
a_{21}B & a_{22}B & \cdots a_{2s}B \\
\vdots & \vdots & & \vdots \\
a_{r1}B & a_{r2}B & \cdots & a{rs}B
\end{pmatrix}$$
\end{definition}

Next, we describe the matrices that make up the product. Consider a $p\times p$ matrix $B$ where every $2\times 2$ submatrix satisfies the property in (\ref{blockrelation}). Let the first row of the matrix be $\{1,-\alpha,\alpha^2,-\alpha^3,\ldots,\alpha^{p-1}\}$ and the first column be $\{1,-\alpha^{1q},\alpha^{2q},-\alpha^{3q},\ldots,\alpha^{q(p-1)}\}$. We can use this information to find every other element in the matrix.

\begin{lemma}\label{bformula}
The entry of the matrix $B$ in row $i$ and column $j$ has the form 
\begin{equation}\label{ijformula}
    B_{i,j}=\sum_{n=0}^{\min(i,j)-1}(-1)^{i+j-n}{i-1\choose n}{i+j-n-2 \choose i-1}\alpha^{(i-1-n)q+j-1-n}\gamma^n.
\end{equation}
\end{lemma}

\begin{proof}
We will prove this formula by induction. First, we can easily verify that the first row and first column satisfy this formula. Now, for some $c,d$ integers with $1\leq c< p$ and $1\leq d < p$ assume that entries $B_{c,d}$, $B_{c+1,d}$, and $B_{c,d+1}$ satisfy the formula. We want to show that these imply that $B_{c+1,d+1}$ satisfies the formula. We can calculate the $B_{c+1,d+1}$ entry by applying inductive hypothesis and the recurrence relation (\ref{blockrelation}), yielding

\begin{align*}
    B_{c+1,d+1}&=-\alpha^q\sum_{n=0}^{\min(c-1,d)}(-1)^{c+d+1-n}{c-1\choose n}{c+d-n-1 \choose c-1}\alpha^{(c-1-n)q+d-n}\gamma^n \\
    &\;\;\;-\alpha\sum_{n=0}^{\min(c,d-1)}(-1)^{c+d+1-n}{c\choose n}{c+d-n-1 \choose c}\alpha^{(c-n)q+d-1-n}\gamma^n \\
    & \;\;\;-\gamma\sum_{n=0}^{\min(c-1,d-1)}(-1)^{c+d-n}{c-1\choose n}{c+d-n-2 \choose c-1}\alpha^{(c-1-n)q+d-1-n}\gamma^n. \\
\end{align*}
We reindex the third sum, yielding
\[-\gamma\sum_{n=1}^{\min(c,d)}(-1)^{c+d-n+1}{c-1\choose n-1}{c+d-n-1 \choose c-1}\alpha^{(c-n)q+d-n}\gamma^{n-1}.\]
We then simplify by combining the summations and bringing the negative signs into the sum, increasing the power of (-1). Let $$X_n={c-1\choose n}{c+d-n-1\choose c-1}+{c\choose n}{c+d-n-1\choose c}+{c-1\choose n-1}{c+d-n-1\choose c-1}.$$  This can be simplified differently depending on whether $c>d$, $c<d$, or $c=d$. First, consider that $c\geq d$.
Then substituting the value of $X$ we get
\begin{align*}
    B_{c+1,d+1}&=\sum_{n=1}^{d-1}(-1)^{c+d-n+2}X_n\alpha^{(c-n)q+d-n}\gamma^n \\
    &+(-1)^{c+d+2}\left({c-1\choose 0}{c+d-1\choose c-1}+{c\choose 0}{c+d-1\choose c}\right)\alpha^{cq+d}\gamma^0 \\
    &+(-1)^{c+2}{c-1\choose d}{c-1\choose c-1}\alpha^{(c-d)q}\gamma^d \\
    &+(-1)^{c+2}{c-1\choose d-1}{c-1\choose c-1}\alpha^{(c-d)q}\gamma^d \\
    &=\sum_{n=0}^{c}(-1)^{c+d-n+2}X_n\alpha^{(c-n)q+d-n}\gamma^n.
\end{align*}

Now assume that $c<d$ and use the same definition for $X_n$.
\begin{align*}
    B_{c+1,d+1}&=\sum_{n=1}^{c-1}(-1)^{c+d-n+2}X_n\alpha^{(c-n)q+d-n}\gamma^n \\
    &+(-1)^{c+d+2}\left({c-1\choose 0}{c+d-1\choose c-1}+{c\choose 0}{c+d-1\choose c}\right)\alpha^{cq+d}\gamma^0 \\
    &+(-1)^{d+2}{c\choose c}{d-1\choose c}\alpha^{d-c}\gamma^c \\
    &+(-1)^{d+2}{c-1\choose c-1}{d-1\choose c-1}\alpha^{d-c}\gamma^c \\
    &=\sum_{n=0}^{c}(-1)^{c+d-n+2}X_n\alpha^{(c-n)q+d-n}\gamma^n.
\end{align*}
Both cases result in the simplification
$$B_{c+1,d+1}=\sum_{n=0}^{\min(c,d)}(-1)^{c+d-n+2}X_n\alpha^{(c-n)q+d-n}\gamma^n.$$
Now that we have simplified the expression into one summation, we can simplify the value of $X_n$. Using known identities of binomial coefficients, the binomial coefficients can be rewritten as follows:
$${c+d-1-n\choose c-1}=\frac{c}{c+d-n}{c+d-n\choose c},$$
$${c+d-1-n\choose c}=\frac{d-n}{c+d-n}{c+d-n\choose c},$$
$${c-1\choose n-1}=\frac{n}{c}{c\choose n},$$
$${c-1\choose n}=\frac{c-n}{c}{c\choose n}.$$
Now we will substitute these identities into the expression for $X_n$.
\begin{align*}
    X_n&={c-1\choose n}{c+d-1-n\choose c-1}+{c\choose n}{c+d-1-n\choose c}+{c-1\choose n-1}{c+d-1-n\choose c-1} \\
    &={c\choose n}{c+d-n\choose c}\left(\frac{c-n}{c}\cdot\frac{c}{c+d-n}+\frac{d-n}{c+d-n}+\frac{n}{c}\cdot\frac{c}{c+d-n}\right) \\
    &={c\choose n}{c+d-n\choose c}\left(\frac{c+d-n}{c+d-n}\right) \\
    &={c\choose n}{c+d-n\choose c}.
\end{align*}
Plugging in the simplified value of $X_n$ into the summation for $B_{c+1,d+1}$ yields $$B_{c+1,d+1}=\sum_{n=0}^{d}(-1)^{c+d-n+2}{c\choose n}{c+d-n\choose c}\alpha^{(c-n)q+d-n}\gamma^n,$$
which satisfies the formula. By induction, every entry in $B$ satisfies ~\ref{ijformula}, and the proof is complete.
\end{proof}
The formula for each term allows us to find the exact values of the last row and last column.

\begin{lemma}\label{lastrowcolumn}
The last row of matrix $B$ is $B_{pj}=(-1)^{j-1}\alpha^{(p-j)q}\gamma^{j-1}$ and the last column is $B_{ip}=(-1)^{p-1}\alpha^{p-i}\gamma^{i-1}$.
\end{lemma}

\begin{proof}
The last row of the matrix is given by $$B_{pj}=\sum_{n=0}^{j-1}(-1)^{p+j-n}{p-1\choose n}{p+j-n-2\choose p-1}\alpha^{(p-1-n)q+j-1-n}\gamma^n.$$
By binomial coefficient identities, ${p-1\choose n}\equiv(-1)^n\mod{p}$. Also by binomial coefficient identities,
$${p+j-n-2\choose p-1}=\frac{p}{j-n-1}{p+j-n-2\choose p}.$$
This implies that ${p+j-n-2\choose p-1}\equiv 0\mod{p}$ except when $n=j-1$. Therefore, the only terms that are left in the last row are the terms with $\gamma^{j-1}$. Thus, $B_{pj}=(-1)^{j-1}\alpha^{(p-j)q}\gamma^{j-1}$, proving the lemma for the last row.

The last column of the matrix is given by
$$B_{ip}=\sum_{n=0}^{i-1}(-1)^{p+j-n}{i-1\choose n}{i+p-n-2\choose i-1}\alpha^{(i-1-n)q+p-1-n}\gamma^n.$$
Since $${p+i-n-2\choose p-1}=\frac{(p+i-n-2)!}{(i-1)!(p-n-1)!},$$
then ${p+i-n-2\choose p-1}$ is divisible by $p$ as long as $i-n-2\geq0$. This is the case except when $n=i-1$, because $i-(i-1)-2=-1$. Therefore, all the terms of the entries in the last column become $0$ modulo $p$ except when $n=i-1$. Therefore, the formula for the last column of the matrix is $B_{ip}=(-1)^{p-1}\alpha^{p-i}\gamma^{i-1}$, completing the proof of the lemma.
\end{proof}

Now that we have defined one of the matrices, we can define the sequence of matrices that will be in the Kronecker Product.

\begin{definition}\label{bhdefinition}
Let $p$ be an odd prime.  For $h$ any integer, $1\leq h\leq l$, define $B_h$ to be the $p\times p$ matrix where the $(i,j)$ entry is
$$\left(\sum_{n=0}^{\min(i,j)-1}(-1)^{i+j-n}{i-1\choose n}{i+j-n-2 \choose i-1}\alpha^{(i-1-n)q+j-1-n}\gamma^n\right)^{p^{l-h}}.$$
\end{definition}

\begin{lemma}\label{mainmatrixlemma}
Assume $q=p^l$ for $p$ and odd prime and $l$ a natural number. Then
$$\begin{pmatrix}
P_0 & P_q & \cdots & P_{(q-1)q} \\
P_1 & P_{q+1} & \cdots & P_{(q-1)q+1} \\
\vdots & \vdots & & \vdots \\
P_{q-1} & P_{2q-1} & \cdots & P_{q^2-1}
\end{pmatrix}=B_1\otimes B_2\otimes\cdots\otimes B_l.
$$
\end{lemma}

\begin{proof}
Denote the matrix on the left by $\Gamma_q$ and the matrix on the right side by $\Gamma_q'$. The first row of $\Gamma_q'$ is $(1,-\alpha,\alpha^2,-\alpha^3,\ldots,\alpha^{q-1})$ and the first column is $(1,-\alpha^q,\alpha^{2q},\ldots,\alpha^{(q-1)q})$. Using the definiton of the Kronecker product and Lemmas ~\ref{pkcondtion} and ~\ref{pkrule}, the first rows and first columns of $\Gamma_q$ and $\Gamma_q'$ are equal. Therefore, in order to show that $\Gamma_q=\Gamma_q'$, it is sufficient to show that every $2\times 2$ matrix inside of $\Gamma_q'$ satisfies (\ref{blockrelation}).

We will proceed by induction. First, we know that the matrix $B_l$ satisfies property \ref{blockrelation} by construction and by Lemma ~\ref{bformula}. Now let $i\geq1$ and assume that every $2\times 2$ block of the $p^{l-i}\times p^{l-i}$ matrix
$$C_{i}=B_{i+1}\otimes B_{i+2}\otimes \cdots\otimes B_l$$
satisfies (\ref{blockrelation}). To complete the inductive step, the goal is to show that every $2\times 2$ block of the matrix
$$B_i\otimes C_i=\left(\begin{tabular}{ c | c | c | c}
 $C_i$ & $(-\alpha)^{p^{l-i}}C_i$ & $\cdots$ & $(\alpha^{p-1})^{p^{l-i}}C_i$\\ 
 \hline
$(-\alpha^q)^{p^{l-i}}C_i$ & $(2\alpha^{q+1}-\gamma)^{p^{l-i}}C_i$ & $\cdots$ & $(-\gamma\alpha^{p-2})^{p^{l-i}}C_i$\\  
 \hline
 $\vdots$ & $\vdots$ &   & $\vdots$ \\
 \hline
 $(\alpha^{q(p-1)})^{p^{l-i}}C_i$ & $(-\gamma\alpha^{(p-2)q})^{p^{l-i}}C_i$ & $\cdots$ & $(\gamma^{p-1})^{p^{l-i}}C_i$
\end{tabular}\right)$$
also satisfies (\ref{blockrelation}). There are four cases for where a $2\times 2$ block may lie in the matrix above.
\begin{enumerate}
\item The block lies entirely in one of the $p^2$ cells.
\item The block intersects four different cells of the matrix.
\item The block intersects two horizontally adjacent cells.
\item The block intersects two vertically adjacent cells.
\end{enumerate}
Any $2\times 2$ block in $C_i$ satisfies the relation by the induction hypothesis. Therefore, any $2\times 2$ block $M$ of $B_i\otimes C_i$ in the first case will also satisfy the relation because multiplying by a constant will maintain the relation.

Note that by Lemma ~\ref{lastrowcolumn} and the definition of the Kronecker product, the first and last columns and rows of $C_i$ have the following structure:
\begin{align}\label{C_ifirstlast}
(C_i)_{1,j}=&(-\alpha)^{j-1},\\
(C_i)_{k,1}=&(-\alpha^q)^{k-1},\\
(C_i)_{p^{(l-i)},j}=&(-\alpha^q)^{p^{(l-i)}-j} \gamma^{j-1}\\
(C_i)_{k,p^{(l-i)}}=&(-\alpha)^{p^{(l-i)}-k}\gamma^{k-1}.
\end{align}
For case 2, let $\left(\begin{array}{cc} w & x\\ y & z\end{array}\right)$ be the $2\times 2$ submatrix of $B_i$ corresponding to the four cells of $B_i\otimes C_i$ that the block $M$ intersects.  We know that $w=\tilde{w}^{p^{(l-i)}}$, $x=\tilde{x}^{p^{(l-i)}}$, $y=\tilde{y}^{p^{(l-i)}}$, and $z=\tilde{z}^{p^{(l-i)}}$, where 
\[\tilde{z}=-\alpha^q\tilde{x}-\alpha\tilde{y}-\gamma\tilde{w}.\]
That means
\begin{equation}\label{eq:case2recur}
    z=-(\alpha^q)^{p^{(l-i)}}x-\alpha^{p^{(l-i)}}y-\gamma^{p^{(l-i)}}w.
\end{equation}

Then the $2\times2$ block of $B_i\otimes C_i$ will be of the form
\[M=\left(
\begin{array}{c|c}
\gamma^{p^{(l-i)}-1}w & (-\alpha^q)^{(p^{(l-i)}-1)}x \\ \hline
(-\alpha)^{p^{(l-i)}-1}y & z \\
\end{array}
\right).
\]

Applying equation (\ref{eq:case2recur}), we then have that 
\begin{align*}
    -\alpha^qM_{1,2}-\alpha M_{2,1}-\alpha M_{1,1}&=-\alpha^q(-\alpha^q)^{(p^{(l-i)}-1)}x-\alpha(-\alpha)^{p^{(l-i)}-1}y-\gamma(\gamma^{p^{(l-i)}-1})w\\
    &=(-\alpha^q)^{p^{(l-i)}}x+(-\alpha)^{p^{(l-i)}}y-\gamma^{p^{(l-i)}}w\\
    &=-(\alpha^q)^{p^{(l-i)}}x-\alpha^{p^{(l-i)}}y-\gamma^{p^{(l-i)}}w\\
    &=z=M_{2,2}.
\end{align*} Therefore the $2\times 2$ block $M$ satisfies (\ref{blockrelation}) in case 2.

Now consider case 3, when the block $M$ intersects two horizontally adjacent cells. In this case, the block will have the form
$$M=\left(
\begin{array}{c|c}
(-\alpha)^{p^{(l-i)}-k+1}\gamma^{k-2}x &  (-\alpha^q)^{k-2}y\\ \hline
(-\alpha)^{p^{(l-i)}-k}\gamma^{k-1} x & (-\alpha^q)^{k-1}y \\
\end{array}
\right)$$
for constants $x$ and $y$ and some $2\leq k\leq p^{l-i}$. We see that
\begin{align*}
-\alpha^qM_{1,2}-\alpha M_{2,1}-\alpha M_{1,1}&=-\alpha^q(-\alpha^q)^{k-2}y-\alpha(-\alpha)^{p^{(l-i)}-k}\gamma^{k-1} -\gamma(-\alpha)^{p^{(l-i)}-k+1}\gamma^{k-2}x \\ &= (-\alpha^q)^{k+1}y=M_{2,2}  
\end{align*}
satisfying equation (\ref{blockrelation}). 

Finally, in case 4 of $M$ intersecting vertically adjacent cells, we have
$$M=\left(
\begin{array}{c|c}
(-\alpha^q)^{p^{(l-i)}-j+1} \gamma^{j-2}x & (-\alpha^q)^{p^{(l-i)}-j} \gamma^{j-1}x \\ \hline
(-\alpha)^{j-2}y & (-\alpha)^{j-1}y \\
\end{array}
\right)$$ for some constants $x,y$ and $2\leq j\leq p^{l-i}$.
We see that $M$ satisfies (\ref{blockrelation}) in this case because
\begin{align*}
-\alpha^qM_{1,2}-\alpha M_{2,1}-\alpha M_{1,1}&=-\alpha^q(-\alpha^q)^{p^{(l-i)}-j} \gamma^{j-1}x-\alpha(-\alpha)^{j-2}y -\gamma(-\alpha^q)^{p^{(l-i)}-j+1} \gamma^{j-2}x \\
&=(-\alpha)^{j-1}y = M_{2,2}.
\end{align*}
Therefore, every $2\times 2$ block of $B_i\otimes C_i$ satisfies (\ref{blockrelation}). By induction, every $2\times 2$ block in $\Gamma_q'$ must also satisfy (\ref{blockrelation}). Since the first rows and columns of $\Gamma_q$ and $\Gamma_q'$ are equal, and both satisfy the recurrence relation, then the matrices must be equal and $\Gamma_q=\Gamma_q'$.
\end{proof}

The matrix identity in Lemma \ref{mainmatrixlemma} can be used to find a sufficient condition for $k$ that satisfies $deg_{\alpha,\beta}(t^k)<q$.  The proof of this lemma is structurally identical to Theorem 10 in \cite{malmskog:article}, but uses the results proven here for odd primes.

\begin{lemma}\label{connectpieces}
Let $q=p^l$ for $p$ an odd prime. For an integer $k$, $0\leq k\leq q^2$, write $k=wq+z$ where for non-negative integers $z,q$ with $z<q$. For $\alpha,\beta\in\mathbb{F}_{q^2}$, if $w=0$ or there exists $1\leq i\leq l$ such that $w\equiv 0\mod{p^i}$ and $z\not\equiv -1\mod{p^i}$, then $deg_{\alpha,\beta}(t^k)<q$.
\end{lemma}

\begin{proof}
    Suppose that $k$, $q$, and $z$ satisfy the conditions stated in the lemma where $k=wq+z$. By Lemma \ref{pkcondtion}, we just need to show that $P_{k+1}=-\alpha^qP_k$ in order to show that $deg_{\alpha,\beta}(t^k)<q$.

    When $w=0$, then $k=z<q$. This automatically gives that $\deg_{\alpha,\beta}(t^k)<q$ since $\deg(t^k)<q$.

    Then, say $w>0$ and that there exists an $i$ such that $w\equiv 0\mod{p^i}$ and $z\not\equiv -1\mod{p^i}$. Let
    $$A_i=B_1\otimes\cdots\otimes B_{l-i}\;\;\;\text{and}\;\;\; C_i=B_{l-i+1}\otimes\cdots\otimes B_l,$$
    where $B_h$ is as in Definition \ref{bhdefinition}. By Lemma \ref{mainmatrixlemma}, we have
    \[\begin{pmatrix}
P_0 & P_q & \cdots & P_{(q-1)q} \\
P_1 & P_{q+1} & \cdots & P_{(q-1)q+1} \\
\vdots & \vdots & & \vdots \\
P_{q-1} & P_{2q-1} & \cdots & P_{q^2-1}
\end{pmatrix}=\begin{pmatrix}
a_{11}C_i & a_{12}C_i & \cdots & a_{1s}C_i \\
a_{21}C_i & a_{22}C_i & \cdots & a_{2s}C_i \\
\vdots & \vdots & & \vdots \\
a_{s1}C_i & a_{s2}C_i & \cdots & a_{ss}C_i
\end{pmatrix}=A_i\otimes C_i,\]
where $s=p^{l-i}$.

First assume that $P_k$ lies in block $a_{cd}C_i$ for some $c,d\in\{1,\ldots,p^{l-i}\}$. Because $w\equiv 0\mod{p^i}$, then $P_k$ is in the first column of $a_{cd}C_i$. Also, since $z\not\equiv -1\mod{p^i}$ then $P_k$ is not in the last row of $a_{cd}C_i$. Therefore, $P_{k+1}$ must also be in the same block $a_{cd}C_i$. Since the first column of $C_i$ is as given in (\ref{C_ifirstlast}), we get that $P_{k+1}=-\alpha^qP_k$. Therefore by Lemma \ref{pkcondtion}, $\deg_{\alpha\beta}(t^k)<q$.
\end{proof}

\subsection{Counting the number of Good Monomials}

So far, we have proved that monomials of the form $M_{a,b}(x,y)=x^ay^b$ for $a\leq q-1$ and $b\leq q^2-1$ give rise to a linearly independent set of evaluation vectors, and have found a sufficient condition for when $deg_{\alpha,\beta}(t^k)<q$. Now, we want to count the monomials that fit the condition. The coefficients of powers of $t$ in $M_{a,b}(L_{\alpha,\beta}(t))$ can be found by expanding: 
\[M_{a,b}(L_{\alpha,\beta}(t))=M_{a,b}(\alpha t+\beta,t)=(\alpha t+\beta)^at^b=\sum_{j=0}^a{a\choose j}\alpha^j\beta^{a-j}t^{b+j}.\]
If $a+b<q$, then it is clear the $M_{a,b}$ is good. If $a+b\geq q$, there are two ways that $M_{a,b}$ can be good, as described in Section 3.2 of \cite{malmskog:article} and included here for completeness. The first mechanism is that all the terms could reduce to degree less than $q$ modulo $p_{\alpha,\beta}(t)$ without using finite field properties. Second, the coefficient in front of the term $t^q$ could reduce to 0 modulo $p$. Both of these contribute to good monomials.  

To understand when the binomial coefficients will vanish modulo $p$, we use Lucas' Theorem.

\begin{definition}
Let $c$ and $d$ be integers between $0$ and $p^k-1$ for prime $p$, with $c\leq d$ and let $p-ary(c)\in\{0,1,\ldots,p-1\}^k$ denote the p-ary expansion of $c$ (also for $d$). We say that $c$ lies in the $p$-shadow of $d$, denoted $c\leq_pd$ if every digit of $p-ary(c)$ is less than or equal to the corresponding digit in $p-ary(d)$.
\end{definition}

\begin{theorem}
(Lucas). Let $0\leq c \leq d$ be integers. Then ${d\choose c}$ is nonzero mod $p$ if and only if $c\leq_p d$, i.e. $c$ lies in the $p$-shadow of $d$.
\end{theorem}

With these tools, we now find some sufficient conditions on $a$ and $b$ for good $M_{a,b}$ and count the number of monomials that satisfy those conditions to get a lower bound on the rate of the Hermitian-lifted Code. Note that the properties in the following lemma do not cover all possible good monomials, but are plentiful enough to provide a positive lower bound. The following lemma and proof are identical to Claim 12 from \cite{malmskog:article}, with a general prime $p$ replacing every $2$ from the original version. We define $x_r$ to be the digit corresponding to $p^r$ of the $p$-ary expansion of $x$, or the digit $r$ positions from the right. We include this here for completeness.

\begin{lemma}\label{goodproperties}
\normalfont(\cite{malmskog:article}) Suppose that $a\leq q-1$ and $b\leq q^2-1$ satisfy the following:
\begin{enumerate}
    \item $b=wq+b'$ for some $w<q$ and some $b'<p^{l-1}$ so that $w\equiv 0\mod(p^i)$ for some $1\leq i\leq l$,
    \item $a<p^{l-1}$,
    \item there is some $0\leq s\leq i-1$ so that $a_s=b_s'=0$.\end{enumerate}
Then $M_{a,b}$ is good.
\end{lemma}
\begin{proof}
    If $a,b$ satisfy the above, then 
    \begin{equation}
        M_{a,b}(L_{\alpha,\beta}(t))=\sum_{j=0}^a{a\choose j}\alpha^j\beta^{a-j}t^{b+j}=\sum_{j\leq_p a}{a\choose j}\alpha^j\beta^{a-j}t^{b+j}
    \end{equation} by Lucas' theorem.  For any term here where the binomial coefficient $a\choose j$ did not vanish modulo $p$, we have $j\leq a<p^{l-1}$ and $b_s=j_s=0$ for some $s<i$ as in conditions 2 and 3 above.  So each non-vanishing term involves $t^k$ where $k=wq+b'+j$ and $j\leq_p a$. By condition 1, we have $w\equiv 0\mod(p^i)$, and we now consider $b'+j \mod p^i$.  Write $b'=b''p^{s+1} +b'''$ and $j=j''p^{s+1}+j'''$, where $b''', j'''<p^s$.  Let $2^{s+1}(b''+j'')\equiv c \mod p^i$ with $0\leq c\leq p^i-1$.  We then have $c\leq p^i-p^{s+1}$ because $c$ is equivalent to a multiple of $p^{s+1}$.  So
    \[ b'+j\equiv c+b'''+j''' \mod p^i,\] where \[c+b'''+j'''< p^i-p^{s+1}+2p^s<p^i-p^{s+1}+p^{s+1}-1=p^i-1.\]  Thus $b'+j\not\equiv -1 \mod p^i$.  That means $k=wq+z$ where $w\equiv 0 \mod p^i$ and $z\not\equiv -1 \mod p^i$.  By Lemma \ref{connectpieces}, $\deg_{\alpha,\beta}(t^k)<q$ for all $\alpha,\beta$.  Therefore $\deg_{\alpha,\beta}(M_{a,b}(L_{\alpha,\beta}(t)))<q$, so $M_{a,b}$ is good.
    \end{proof}

We are now able to prove our main theorem:
\begin{theorem}\label{main}
   Suppose that $q=p^l$ where $p$ is an odd prime and $l\geq 2$. Then the rate of $\mathcal{C}_q$ is at least $$\frac{.469}{p^4(p-1)(p^3-p^2+1)}.$$
\end{theorem}

\begin{proof}
We will count the number of pairs $a,b$ that satisfy the sufficient conditions from Lemma \ref{goodproperties} for $M_{a,b}$ to be good. We iterate over all $s$, where we take $s$ to be the smallest index so that $a_s=b_s'=0$. For a given $s$, there are $p^{2s}-(p^2-1)^s$ ways to assign the bits $a_0,\ldots,a_{s-1}$ and $b_0',\ldots,b_{s-1}'$, since there are only $(p^2-1)^s$ ways to never have $a_r=b_r'=0$ for any $0\leq r\leq s-1$. Then there are $p^{2(l-s-2)}$ ways to assign the bits $a_{s+1},\ldots,a_{l-2},b_{s+1}',\ldots,b_{l-2}'$. Finally, there are $p^{l-s-1}$ ways to assign the bits $w_{s+1},\ldots,2_{l-1}$. We will choose $w_0=w_1=\ldots=w_s=0$, ensuring that $w\equiv 0\mod p^{s+1}$, specifically $w\equiv 0\mod{p^i}$ for some $i>s$. Thus, the total number of monomials meeting the description in Lemma \ref{goodproperties} when $l\geq 2$ is
\begin{align*}
    \sum_{s=0}^{l-1}&(p^{2s}-(p^2-1)^s)p^{2(l-s-2)}p^{l-s-1}  =\sum_{s=0}^{l-1}(p^{2s}-(p^2-1)^s)p^{2(l-s-2)+l-s-1} \\
    & =\sum_{s=0}^{l-1}(p^{2s}-(p^2-1)^s)p^{3l-3s-5} \\
    & = \frac{p^{3l}}{p^5}\sum_{s=0}^{l-1}(p^{2s}-(p^2-1)^s)p^{-3s} \\
    & =\frac{p^{3l}}{p^5}\sum_{s=0}^{l-1}\left(\frac{1}{p}^s-\left(\frac{p^2-1}{p^3}\right)^s\right) \\
    & = \frac{p^{3l}}{p^5}\left(\frac{1-(\frac{1}{p})^l}{1-\frac{1}{p}}-\frac{1-(\frac{p^2-1}{p^3})^l}{1-\frac{p^2-1}{p^3}}\right)\\
    & = \frac{p^{3l}}{p^5}\left(\frac{p}{p-1}\left(1-\left(\frac{1}{p}\right)^l\right)-\frac{p^3}{p^3-p^2+1}\left(1-\left(\frac{p^2-1}{p^3}\right)^l\left(\frac{1}{p}\right)\right)\right) \\
    & = \frac{p^{3l}}{p^5}\left(\frac{p}{p-1}-\frac{p}{p-1}\left(\frac{1}{p}\right)^l-\frac{p^3}{p^3-p^2+1}+\frac{p^3}{p^3-p^2+1}\left(\frac{p^2-1}{p^3}\right)^l\right) \\
    & =\frac{p^{3l}}{p^5}\left(\frac{p}{(p-1)(p^3-p^2+1)}-\frac{p}{p-1}\left(\frac{1}{p}\right)^l+\frac{p^3}{p^3-p^2+1}\left(\frac{p^2}{p^3-p^2+1}\right)^l\right) \\
    &= \frac{p^{3l}}{p^5}\left(\frac{p}{(p-1)(p^3-p^2+1)}\right)(1-(p^3-p^2+1)\left(\frac{1}{p}\right)^l+(p^3-p^2)\left(\frac{p-1}{p^3}\right)^l)\\
    &=q^3\left(\frac{1-(p^3-p^2+1)(\frac{1}{p})^l+(p^3-p^2)(\frac{p^2-1}{p^3})^l}{p^4(p-1)(p^3-p^2+1)}\right) \\
    & \geq q^3\left(\frac{1-.531}{p^4(p-1)(p^3-p^2+1)}\right)
\end{align*}
The last step comes from finding the lower bound on the numerator when $p\geq 3$ and $l\geq 2$. The rate of the code is $r=k/n$ where $n=q^3$. Therefore, we have that the rate when $p$ is an odd prime is bounded below by
$$\frac{.469}{p^4(p-1)(p^3-p^2+1)},$$
which completes the proof.
\end{proof}

\section{Conclusion}

This paper proved an extension of the main theorem in \cite{malmskog:article} by following a similar proof strategy to conclude that all Hermitian-lifted Codes have a rate bounded below by a positive constant, regardless of the value of $q$.

There are remaining unanswered questions regarding Hermitian-lifted codes and also lifted codes in general. In \cite{undergrad:future}, the authors improve the lower bound given in \cite{malmskog:article} and, in work completed after this work, improve upon Theorem \ref{main} by using a different proof strategy for counting good monomials. However, this work still does not find the exact dimension of the Hermitian-lifted code. Similar questions could also be studied for lifted codes on different types of curves, in the direction of \cite{matthews_murphy} and \cite{matthews2023curvelifted}.


\end{document}